\documentclass{statsoc}

\usepackage[a4paper]{geometry}
\usepackage{graphicx}
\usepackage[textwidth=8em,textsize=small]{todonotes}
\usepackage{amsmath}
\usepackage{natbib}

\usepackage{booktabs}
\usepackage{natbib}
\usepackage{amsbsy}
\usepackage{amssymb}
\usepackage{amsfonts,amsmath}
\usepackage{array}
\usepackage{float}
\usepackage{graphicx}
\usepackage{hyperref}
  \PassOptionsToPackage{hyphens}{url}\usepackage{hyperref}
\usepackage{pdflscape}
\usepackage{multirow}
\usepackage{wasysym}
\usepackage{longtable}
\usepackage{amssymb}
\usepackage{subfigure}
\usepackage[figuresright]{rotating}
\usepackage{color}
\usepackage{soul}
\setstcolor{red}

\usepackage{bm}
\usepackage{amsfonts}
\usepackage{algorithm} 
\usepackage{algpseudocode} 
\usepackage{tabularx}
\usepackage{enumitem}
\usepackage{amsmath}
 \usepackage{float}
\usepackage{graphicx} 
\usepackage{color} 
\usepackage{subfigure} 
\newtheorem{thm}{Theorem}

\title[A Novel Mixture Model for Characterizing Human Aiming Performance Data]{A Novel Mixture Model for Characterizing Human Aiming Performance Data}
\author[Author 1 {\it et al.}]{Yanxi Li}
\address{Dr.~Bing Zhang Department of Statistics, University of Kentucky, Lexington, Kentucky 40509, USA}
\author{Derek S. Young}
\address{Dr.~Bing Zhang Department of Statistics, University of Kentucky, Lexington, Kentucky 40509, USA}
\email{derek.young@uky.edu}
\author{Julien Gori}
\address{ISIR, CNRS UMR 7222, Sorbonne Université, Paris, France}
\author[Y. Li et al.]{Olivier Rioul}
\address{Telecom Paris, Institut Polytechnique de Paris, France}

\begin{document}
\begin{abstract}
Fitts' law is often employed as a predictive model for human movement, especially in the field of human-computer interaction.  Models with an assumed Gaussian error structure are usually adequate when applied to data collected from controlled studies.  However, observational data (often referred to as data gathered ``in the wild") typically display noticeable positive skewness relative to a mean trend as users do not routinely try to minimize their task completion time.  As such, the exponentially-modified Gaussian (EMG) regression model has been applied to aimed movements data.  However, it is also of interest to reasonably characterize those regions where a user likely was not trying to minimize their task completion time.  In this paper, we propose a novel model with a two-component mixture structure -- one Gaussian and one exponential -- on the errors to identify such a region.  An expectation-conditional-maximization (ECM) algorithm is developed for estimation of such a model and some properties of the algorithm are established.  The efficacy of the proposed model, as well as its ability to inform model-based clustering, are addressed in this work through extensive simulations and an insightful analysis of a human aiming performance study.
\end{abstract}

\keywords{block relaxation, \and ECM algorithm, \and exponentially-modified Gaussian, \and Fitts' law, \and human-computer interaction, \and model-based clustering}

\section{Introduction}
\label{sec:intro}

An individual's reaction time and movement time are important markers about the status of their neurological system.  Neurologists believe that reaction time is, perhaps, the most widely-used measure in neuroscience and psychology for noninvasively assessing processing in the brain \citep{eLife}.  For instance, patients with Parkinson's disease are found to have prolonged reaction time and movement time \citep{Parkinson}, while slowed reaction time has also been regarded as an early feature of Alzheimer's disease \citep{Alzheimer}.  

Fitts' law \citep{Fitts} is an empirical law that describes movement time for human voluntary movement. In the Human-Computer Interaction (HCI) field, it has been adapted to model selection times in graphical user interfaces (GUIs) \citep{MacKenzie}.  Specifically, the minimum movement time $t$ needed to select a rectangular target located at a distance $d$ away, with width $w$ and height $h$, is given by $ t = a + b \log_2 \bigg(1 + \frac{d}{\text{min}(h,w)}\bigg) $. This relationship has been demonstrated in numerous GUI contexts through controlled experiments, where participants have been asked to maximize their movement performance by going ``as quickly and precisely" as possible. Usually,  movement time data collected this way has relatively low variance, and the parameters of the linear model, $a$ and $b$, are directly estimated using maximum likelihood estimation (or equivalently, ordinary least squares).

However, Fitts' model is often ill-fitting when the data arises from non-controlled settings, such as crowdsourced web-experiments \citep{Goldberg} or field studies \citep{Chapuis}. It was recently argued that in non-controlled settings, Fitts' law should be interpreted as a model of minimum observed times \citep{Julien1,Julien3}. The idea is that, in these studies, one cannot control for perturbation or participant motivation, which may increase (but not decrease) the movement time needed to select the target. At the same time, one cannot na\"ively fit a lower bound to the dataset by identifying minimum movement times, since some movements may be poorly segmented. Another opportunity for lower than possible movement times is when participants accidentally click on a target (involuntary movement). It was previously shown that in this case, Fitts' law could be recovered using an exponentially-modified Gaussian (EMG) regression \citep{Julien2}.

The EMG distribution is defined as a convolution of the distributions of two independent random variables, where one follows a Gaussian distribution and the other follows an exponential distribution. A random variable $X$ follows follows an EMG distribution if the density has the form
 \begin{equation}\label{emg}
     f(x;\mu, \sigma, \alpha) = \frac{\alpha}{2}\exp\left\{\frac{\alpha}{2}(2\mu+\alpha \sigma^2-2x)\right\}\operatorname{erfc}\Big( \frac{\mu + \alpha \sigma^2-x}{\sqrt{2}\sigma} \Big),
 \end{equation}
where $\mu \in \mathbb{R}$ and $\sigma^2$ are the variance and mean, respectively, of the Gaussian component, $\alpha>0$ is the rate of exponential component, and $\operatorname{erfc}(\cdot)$ is the complementary error function. We will write $X\sim EMG(\mu, \sigma, \alpha)$ to denote when a random variable follows the EMG distribution as defined above.  Due to its characteristic positive skew from the exponential component, the EMG distribution has provided insight into applied problems across a diverse cross-section of fields, such as microarray preprocessing \citep{Silver},  cell biology \citep{Golubev}, chromatography \citep{Yuri}, and neuropsychology \citep{Palmer}. In the present study on human aiming performance, we seek a more critical examination of the data, which begins with analyzing EMG regression fits for individual subjects.  We seek additional flexibility to understand from which process the individual's performance arises: the one characterized by the Gaussian distribution or the one characterized by the exponential distribution.  The EMG distribution does not allow for classifying such an individual observation, so we propose a competing regression model where the error structure is assumed to be a two-component mixture of a Gaussian and an exponential distribution.  Thus, both the EMG regression model and our novel mixture model are able to characterize data with positive residuals relative to a mean trend, but the latter can also serve to perform model-based clustering. Such clustering results can then assist researchers  attempting to identify outliers produced by technical or human errors. Classifying such outliers remains a topic receiving close attention in the HCI field \citep{cairns_2019}.

We must also address some computational challenges of the two models in this work.  For estimating the EMG regression model, we have found the existing computational routines to not be particularly robust, especially for large datasets like those analyzed in this work.  We develop a block-relaxation algorithm for estimating an EMG regression model with (potentially) multiple predictors.  We then develop an expectation-conditional-maximization \citep[ECM;][]{Meng} algorithm for estimating our novel mixture-of-regressions model. A computational advantage of our novel mixture-of-regressions model is the global convergence of its corresponding ECM algorithm, which we lack in the block-relaxation algorithm for the EMG regression model.

The rest of this paper is organized as follows. In Section \ref{s:model}, we introduce the EMG regression model and our novel mixture-of-regressions model, which we refer to as a \textit{mixture-of-regressions model with flare}, or \textit{flare regression model}, in short.  In Section \ref{s:inf}, we detail the algorithms used for estimating both the EMG regression model and the flare regression model.  We further establish some theoretical properties of the block-relaxation algorithm for estimating the EMG regression model, and prove the global convergence of the ECM algorithm for estimating the flare regression model. Estimation of standard errors for the estimated model parameters  and details about a model-based clustering strategy using the flare regression model are also addressed. In this same section, two simulation studies are performed.  First, a brief numerical  study is conducted to compare the estimation precision of the two algorithms, as well as to demonstrate  model-based clustering using the flare regression model. Second, a large simulation study is performed to assess the robustness of the ECM algorithm and the general efficacy of the flare regression model. In Section \ref{s:data}, we analyze human aiming performance data.  We emphasize the results from the flare regression model, which are benchmarked against the EMG regression results.  Other candidate models are considered in our analysis, but the metrics used demonstrate superior performance of the flare regression model in the presence of more extreme positive residuals. Finally, we conclude with a summary of the main results in Section \ref{s:discuss}.

\section{The Models}
\label{s:model}
For both of the models that we present, let $Y_1,\ldots,Y_n$ denote a random sample of size $n$, where each of these univariate random variables is measured with a vector of $p$-dimensional predictors, $p\in \mathbb{N}^{+}$, given by $\mathbf{X}_1,\ldots,\mathbf{X}_n$.  We use the convention that $X_{i,1}\equiv 1$, $i=1,\ldots,n$, to reflect an intercept in our models. We further let $(y_i,\mathbf{x}_i)$ denote the realizations of the pairs $(Y_i,\mathbf{X}_i)$.  Thus, our focus will be on linear regression models of the form
\begin{equation}\label{linreg}
y_i = \mathbf{x}_i \boldsymbol{\beta} + \epsilon_i,
\end{equation}
but where non-traditional (i.e., non-Gaussian) distributional structures of $\epsilon_i$ will be explored.

First we consider the EMG regression setting, where the error structure for the model in (\ref{linreg}) is $\epsilon_i \sim EMG(0, \sigma, \alpha) $,  for $i=1, \ldots, n$. We next consider the model where $\epsilon_i \sim \lambda \mathcal{N}(0, \sigma^2) + (1-\lambda) Exp(\alpha) $, for $i=1, \ldots, n$.  For this two-component mixture structure on the error terms, $\lambda \in [0, 1]$ is the mixing proportion, $\mathcal{N}(0, \sigma^2)$ is the Gaussian distribution with mean $0$ and variance $\sigma^2$, and $Exp(\alpha)$ is the exponential distribution with mean $\alpha^{-1}$, where $\alpha > 0$ is the rate parameter. This structure gives us the model we refer to as a mixture-of-regressions model with flare, or simply a flare regression model. The etiology of the term ``flare" for our purposes comes from the phenomenon that occurs in gamma-ray bursts, where flaring is an erratic emission of a huge amount of energy on a relatively short timescale \citep{refId0}.  From a data perspective, this behavior manifests as an overall (piecewise) linear trend between the response and predictor(s), but a subset of the data clearly deviates more substantially from the linear trend than the rest of the data. One may also envision this as a form of one-sided contamination.  Scatterplots of simulated data from an EMG regression model and a flare regression model are given in Web Figures 1 and 2, respectively, of the Supporting Information file.

\section{Algorithms, Estimation, and Some Properties}
\label{s:inf}
\subsection{EMG Regression and a Block-Relaxation Algorithm}
Following (\ref{emg}), the density function for the EMG regression model is
\[ f(y_i; \mathbf{x}_i, \boldsymbol{\psi}) = \frac{\alpha}{2}\exp\left\{\frac{\alpha}{2}[\alpha \sigma^2-2(y_i - \mathbf{x}_i^{\top} \boldsymbol{\beta})]\right\}\operatorname{erfc}\Big( \frac{\alpha \sigma^2-(y_i - \mathbf{x}_i^{\top} \boldsymbol{\beta})}{\sqrt{2}\sigma} \Big), \]
which yields the corresponding data loglikelihood
\begin{equation}
    \ell(\boldsymbol{\psi}) = n\Big(\log\frac{\alpha}{2} + \frac{\alpha^2\sigma^2}{2} \Big) -  \sum_{i=1}^n \left\{ \alpha(y_i - \mathbf{x}_i^{\top} \boldsymbol{\beta}) - \log \Big[ \operatorname{erfc}\Big( \frac{\alpha \sigma^2-(y_i - \mathbf{x}_i^{\top} \boldsymbol{\beta})}{\sqrt{2}\sigma} \Big) \Big] \right\}.
\end{equation}
Here, $\boldsymbol{\psi} = ( \boldsymbol{\beta}^{\top}, \sigma^2, \alpha)^{\top}$ is the parameter vector of interest.  To estimate $\boldsymbol{\psi}$, we partition it into two blocks, $(\boldsymbol{\psi}_1^{\top},\boldsymbol{\psi}_2^{\top})^{\top}$, where $\boldsymbol{\psi}_1=\boldsymbol{\beta}$ and $\boldsymbol{\psi}_2=(\sigma^2, \alpha)^{\top}$. Maximum likelihood estimation is performed by setting the objective function $\mathbf{Q}(\boldsymbol{\psi})=  \ell(\boldsymbol{\psi})$.  By using the partitioning we defined for $\boldsymbol{\psi}$, we can then apply the iterative block-relaxation algorithm of \citet{de} to estimate $\boldsymbol{\psi}$. See \textbf{Algorithm 1} in the Appendix for additional details.

The following theorem about concavity properties of EMG models allows us to make some comments about the concavity of $\mathbf{Q}(\boldsymbol{\psi})$.
\begin{thm}\label{thm1}
Let $Y|\mathbf{X}\sim EMG(\mathbf{X}^{\top}\boldsymbol{\beta}, \sigma, \alpha)$ be (conditionally) an EMG random variable.
\begin{enumerate}[leftmargin=5\parindent]
    \item The logarithm of $f(y; \mathbf{x}, \boldsymbol{\psi})$ is strictly concave in $y$.
    \item $\ell(\boldsymbol{\psi})$ is strictly concave in $\boldsymbol{\beta}$.
    \item $\ell(\boldsymbol{\psi})$ is strictly concave in $\alpha$ if $\alpha \sigma < 1$.
\end{enumerate}
\end{thm}
See Web Appendices A--C in the Supporting Information for a detailed proof of the above.

From Theorem \ref{thm1}, we conclude the objective function $\mathbf{Q}(\boldsymbol{\psi})=  \ell(\boldsymbol{\psi})$ is always strictly concave in $\boldsymbol{\beta}$ and is strictly concave in $\alpha$ if $\alpha \sigma < 1$. It is still challenging to derive sufficient conditions that ensure the strict concavity of $\mathbf{Q}(\boldsymbol{\psi})$ in $\sigma$. Hence, getting sufficient conditions that ensure the strict concavity of $\mathbf{Q}(\boldsymbol{\psi})$ in $\boldsymbol{\psi} $ remains an open problem. Due to the lack of concavity of $\mathbf{Q}(\boldsymbol{\psi})$ in $\boldsymbol{\psi} $, the global convergence of the block-relaxation  algorithm cannot be guaranteed. Moreover, due to the lack of a closed-form expression for the maximum likelihood estimator of $\boldsymbol{\psi}$, we must appeal to numerical optimization methods, like Nelder–Mead and Quasi-Newton, to estimate $\boldsymbol{\psi}$. However, these problems are circumvented in the flare mixture regression setting, which we show after developing the corresponding objective function for estimating its parameters. 

\subsection{The Flare Regression Model and an ECM Algorithm}\label{subsec:flare}
For the flare regression model, the density function is
\begin{equation}\label{flaredens}
\begin{aligned}
    f(y_i; \mathbf{x}_i, \boldsymbol{\theta}) &=\frac{\lambda}{\sqrt{2 \pi \sigma^2}}\exp \Big\{ -\frac{1}{2\sigma^2} (y_i - \mathbf{x}_i^{\top} \boldsymbol{\beta})^2 \Big\} \\
    & \quad \quad \quad \quad + (1-\lambda) \alpha \exp \Big\{-\alpha (y_i - \mathbf{x}_i^{\top} \boldsymbol{\beta}) \Big\}I \Big\{ (y_i - \mathbf{x}_i^{\top} \boldsymbol{\beta}) >0\Big\},
\end{aligned}
\end{equation}
which yields the corresponding (observed) data loglikelihood
\begin{equation}\label{obs}
\begin{aligned}
    \ell_o(\boldsymbol{\theta}) &= \sum_{i=1}^n \log \Bigg\{ \frac{\lambda}{\sqrt{2 \pi \sigma^2}}\exp \Big\{ -\frac{1}{2\sigma^2} (y_i - \mathbf{x}_i^{\top} \boldsymbol{\beta})^2 \Big\}  \\
    & \quad \quad \quad \quad + (1-\lambda) \alpha \exp \Big\{-\alpha (y_i - \mathbf{x}_i^{\top} \boldsymbol{\beta}) \Big\}I \Big\{ (y_i - \mathbf{x}_i^{\top} \boldsymbol{\beta}) >0\Big \} \Bigg\}.
\end{aligned}
\end{equation}
Here, $\boldsymbol{\theta} = (\lambda, \boldsymbol{\beta}^{\top}, \sigma^2, \alpha )^{\top}$ is the parameter vector of interest. Note, however, that finding $\hat{\boldsymbol{\theta}}$ by simply using (\ref{obs}) is challenging as in most finite mixture models, so we consider the $(y_i,\mathbf{x}_i)$ as incomplete data resulting from non-observed complete data.  The data is made complete by augmenting the problem with the unobserved indicators $Z_i = I\{$observation $i$ belongs to the Gaussian component$\}$.  Thus, the complete-data loglikelihood is easily found to be
\begin{equation}\label{compdata}
\begin{aligned}
     \ell_c(\boldsymbol{\theta}) = \sum_{i=1}^n & \Bigg[ Z_i \log \left( \frac{\lambda}{\sqrt{2 \pi \sigma^2}} \exp \Big\{ -\frac{1}{2 \sigma^2} (y_i - \mathbf{x}_i^{\top} \boldsymbol{\beta})^2  \Big\} \right)\\
     & + (1 - Z_i)\log\Big( (1-\lambda)\alpha \exp \Big( -\alpha (y_i - \mathbf{x}_i^{\top} \boldsymbol{\beta})I \Big\{ (y_i - \mathbf{x}_i^{\top} \boldsymbol{\beta}) >0\Big\} \Big) \Big)  \Bigg].
\end{aligned}
\end{equation}

Maximum likelihood estimation for finite mixture models is typically performed via an expectation-maximization (EM) algorithm \citep{Dempster}. In many classic parametric mixtures, solutions of the maximization-step (M-step) exist in closed form; see \citet{mcpeel}.  However, we cannot directly estimate $\boldsymbol{\theta}$ for the flare regression model using the above complete-data setup. In particular, we lack a closed-form solution of the regression coefficient vector $\boldsymbol{\beta}$ in the M-step.  However, we mitigate this issue by implementing an iterative procedure within a conditional-maximization-step (CM-step) of an ECM algorithm.

In the first expectation-step (E-Step) for iteration $t$, $t = 0, 1, \ldots,$ we compute the expected complete-data loglikelihood as
\begin{align*}
     \mathbf{Q}(\boldsymbol{\theta}; \boldsymbol{\theta}^{(t)}) = \sum_{i=1}^n & \Bigg[ Z_i^{(t)} \log \Big( \frac{\lambda}{\sqrt{2 \pi \sigma^2}} \exp \Big\{ -\frac{1}{2 \sigma^2} (y_i - \mathbf{x}_i^{\top} \boldsymbol{\beta})^2  \Big\} \Big)\\
     & + (1 - Z_i^{(t)})\log\Big( (1-\lambda)\alpha \exp \Big( -\alpha (y_i - \mathbf{x}_i^{\top} \boldsymbol{\beta})I \Big\{ (y_i - \mathbf{x}_i^{\top} \boldsymbol{\beta}) >0\Big\} \Big)  \Big)  \Bigg],
\end{align*}
where
\begin{equation}
    Z_i^{(t)} = \frac{\frac{\lambda^{(t)}}{2\pi\sigma^{2(t)}} \exp \Big\{ -\frac{1}{2\sigma^{2(t)}} (y_i - \mathbf{x}_i^{\top} \boldsymbol{\beta}^{(t)})^2  \Big\}}{f(y_i; \mathbf{x}_i, \boldsymbol{\theta}^{(t)})}
\end{equation}
is the posterior membership probability of observation $i$ belonging to the Gaussian component of the flare regression model and the denominator is the flare regression density in (\ref{flaredens}). We then partition $\boldsymbol{\theta}$ into $(\boldsymbol{\theta}_1^{\top},\boldsymbol{\theta}_2^{\top})^{\top}$, where $\boldsymbol{\theta}_1=\boldsymbol{\beta}$ and $\boldsymbol{\theta}_2=(\sigma^2, \alpha, \lambda)^{\top}$. 

For the first CM-step, we calculate $\boldsymbol{\theta}_1^{(t+1)} = \underset{\boldsymbol{\theta}_1}{\arg\max} \ \mathbf{Q}(\boldsymbol{\theta}; \boldsymbol{\theta}^{(t)})$.  In this step, we are updating $\boldsymbol{\beta}^{(t+1)}$ by maximizing the objective function $m(\boldsymbol{\beta}) = \mathbf{Q}(\boldsymbol{\beta}; \boldsymbol{\theta}^{(t)})$ with respect to $\boldsymbol{\beta}$, subject to the linear inequality constraints $(1 - Z_i^{(t)})(y_i - \mathbf{x}_i^{\top} \boldsymbol{\beta}) \geq 0$ for $i = 1, \ldots, n$. Not surprisingly, it is challenging to calculate the closed form for the maximum likelihood estimate (MLE) of $\boldsymbol{\beta}$. Instead, we iteratively update $\boldsymbol{\beta}$ using a gradient algorithm introduced by \citet{Lange}:
\begin{equation}
    \boldsymbol{\beta}^{(t+1)} = \boldsymbol{\beta}^{(t)} -  \left[\frac{d^2m}{d\boldsymbol{\beta}^2}\right] ^{-1}\Biggl|_{\boldsymbol{\beta}=\boldsymbol{\beta}^{(t)}}   \frac{dm}{d\boldsymbol{\beta}}\Biggl|_{\boldsymbol{\beta}=\boldsymbol{\beta}^{(t)}},
\end{equation}
where
\begin{equation}
\begin{aligned}
    \frac{dm}{d\boldsymbol{\beta}} &= \sum_{i=1}^n \left[  \frac{Z_i^{(t)}}{\sigma^{2^{(t)}}} \mathbf{x}_i (y_i - \mathbf{x}_i^{\top} \boldsymbol{\beta}) + \alpha^{(t)}(1 - Z_i^{(t)}) \mathbf{x}_i  \right] \ \ \ \text{and} \\
    \frac{d^2m}{d\boldsymbol{\beta}^2} &= -\sum_{i=1}^n \frac{Z_i^{(t)}}{\sigma^{2^{(t)}}} \mathbf{x}_i \mathbf{x}_i^{\top}.
\end{aligned}
\end{equation}


Next, set $\boldsymbol{\theta}^{(t+1/2)} = (\boldsymbol{\theta}_1^{(t+1)\top}, \boldsymbol{\theta}_2^{(t)\top})^{\top}$ for the second E-Step of the current iteration, and obtain the updated posterior membership probabilities as 
\begin{equation}
    Z_i^{(t+1/2)} = \frac{\frac{\lambda^{(t)}}{2\pi\sigma^{2(t)}} \exp \Big\{ -\frac{1}{2\sigma^{2(t)}} (y_i - \mathbf{x}_i^{\top} \boldsymbol{\beta}^{(t+1)})^2  \Big\}}{f(y_i; \mathbf{x}_i, \boldsymbol{\theta}^{(t+1/2)})}, 
\end{equation}
With $\boldsymbol{\theta}_1$ fixed at $\boldsymbol{\theta}_1^{(t+1)}$, we find $\boldsymbol{\theta}_2^{(t+1)} = \underset{\boldsymbol{\theta}_2}{\arg\max} \ \mathbf{Q}(\boldsymbol{\theta}; \boldsymbol{\theta}^{(t)})$, which yields the following MLEs that are weighted using the updated posterior membership probabilities $Z_i^{(t+1/2)}$, $i=1,\ldots,n$:
\begin{align}
    \lambda^{(t+1)} &= \frac{1}{n} \sum_{i=1}^n Z_i^{(t+1/2)} \label{mle1}\\
    \sigma^{2(t+1)} &= \frac{\sum_{i=1}^n Z_i^{(t+1/2)}(y_i - \mathbf{x}_i^{\top} \boldsymbol{\beta}^{(t+1)})^2}{\sum_{i=1}^n Z_i^{(t+1/2)}}  \ \ \ \ \ \ \text{and}  \label{mle2} \\
    \alpha^{(t+1)} &= \frac{\sum_{i=1}^n (1-Z_i^{(t+1/2)})}{\sum_{i=1}^n (1-Z_i^{(t+1/2)}) (y_i - \mathbf{x}_i^{\top} \boldsymbol{\beta}^{(t+1)})}.\label{mle3}    
\end{align}
See \textbf{Algorithm 2} in the Appendix for additional details.


Letting $\boldsymbol{\theta}^{(\infty)}$ and $Z_i^{(\infty)}$, $i=1,\ldots,n$, denote, respectively, the parameter estimates and posterior membership probabilities obtained upon convergence of \textbf{Algorithm 2}, we proceed to set $\widehat{\boldsymbol{\theta}}$ = $\boldsymbol{\theta}^{(\infty)}$ as our estimate for $\boldsymbol{\theta}$.  Moreover, the $Z_i^{(\infty)}$ and $1-Z_i^{(\infty)}$ are the probabilities that an observation's error term came from, respectively, the Gaussian component or the exponential component.  A decision rule can then be defined to determine component membership based on the $Z_i^{(\infty)}$s when compared to a pre-determined cut-off probability $p^*$.  Specifically, the model-based clustering strategy involving our estimated flare regression model is to classify observation $i$ as belonging to the exponential component if $1 - Z_i^{(\infty)} \geq p^*$, otherwise it is classified as belonging to the Gaussian component.  The value used for $p^*$ in our analysis will be discussed later.


Standard errors for mixture models like our flare regression model can be estimated in various ways.  We briefly highlight two ways.  First is to simply bootstrap to obtain the standard errors \citep[see Chapter 2 of][]{mcpeel}. Second is to employ the method due to \citet{Louis}, which calculates the \textit{observed-data information matrix} as the difference between the \textit{complete-data information matrix} and the \textit{missing-data information matrix}; i.e.,
\begin{equation*}
\mathbf{I}(\boldsymbol{\theta}) = -\mathbb{E}_{ \widehat{\boldsymbol{\theta}}} \Bigg( \frac{\partial^2  \ell_c(\boldsymbol{\theta}) }{\partial \boldsymbol{\theta} \partial \boldsymbol{\theta}^{\top}} \Bigg) - \mathbb{E}_{\widehat{\boldsymbol{\theta}}} \Bigg[ \Bigg( \frac{\partial \ell_c(\boldsymbol{\theta})}{\partial \boldsymbol{\theta}} \Bigg) \Bigg( \frac{\partial \ell_c(\boldsymbol{\theta})}{\partial \boldsymbol{\theta}} \Bigg)^{\top} \Bigg] + \mathbb{E}_{ \widehat{\boldsymbol{\theta}}} \Bigg( \frac{\partial \ell_c(\boldsymbol{\theta})}{\partial \boldsymbol{\theta}} \Bigg) \mathbb{E}_{ \widehat{\boldsymbol{\theta}}} \Bigg( \frac{\partial \ell_c(\boldsymbol{\theta})}{\partial \boldsymbol{\theta}} \Bigg)^{\top},
\end{equation*}
where, again, $\ell_c(\boldsymbol{\theta})$ is the complete-data loglikelihood in (\ref{compdata}). Detailed derivations of $\mathbf{I}(\boldsymbol{\theta})$ are in Web Appendix D of the Supporting Information.

\subsubsection{Convergence of the ECM Algorithm}\label{subsec:ECM} \hfill

We will now show the convergence of our ECM algorithm.
For brevity, we denote the above ECM algorithm map as $A(\boldsymbol{\theta})$. Also, denote the update $\boldsymbol{\beta}^{(t+1)}\equiv m(\boldsymbol{\beta}^{(t)})$ and $A(\boldsymbol{\theta}_2; \boldsymbol{\theta}_1)$ as the iterative updating procedure of $A$ assuming a fixed $ \boldsymbol{\theta}_1 =\boldsymbol{\beta} $.

\begin{thm}\label{thm2}
For any fixed $\boldsymbol{\theta}_2$, the gradient updating procedure $m$ converges to a local maximum $\boldsymbol{\beta}^{(\infty)}$.
\end{thm}

\begin{proof}
Assuming $\dim(\boldsymbol{\beta})=p$, $p \in \mathbb{N}^+$, the matrix $d^2m / d\boldsymbol{\beta}^2 $ is negative-definite in every iteration $t$. Also, $m(\boldsymbol{\beta})$ is a continuous concave function in  $\mathbb{R}^p$. Hence, the set $\{ \boldsymbol{\beta} \in \mathbb{R}^p: m(\boldsymbol{\beta}) \geq c \}$ is compact for every constant $c$. The result follows as an immediate consequence of Proposition 1 in \citet{Lange}. $\square$
\end{proof}

\begin{thm}\label{thm3}
For any fixed $\boldsymbol{\theta}_1$, the iterative updating procedure $A(\boldsymbol{\theta}_2; \boldsymbol{\theta}_1)$ converges to a local maximum and $\mathbf{Q}(\boldsymbol{\theta}_2^{(t+1)}; \boldsymbol{\theta}_1) > \mathbf{Q}(\boldsymbol{\theta}_2^{(t)}; \boldsymbol{\theta}_1)$.
\end{thm}

\begin{proof}
Assume a fixed $\boldsymbol{\theta}_2$. Because the updating procedure $m$ converges to the point $\boldsymbol{\beta}^{(\infty)}$, following the result from Proposition 2 in \citet{Lange}, we can conclude that for all sufficiently large $t$, either $\boldsymbol{\beta}^{(t)} = \boldsymbol{\beta}^{(\infty)}$ or $\mathbf{Q}(\boldsymbol{\theta}_1^{(t+1)}; \boldsymbol{\theta}_2) > \mathbf{Q}(\boldsymbol{\theta}_1^{(t)}; \boldsymbol{\theta}_2)$, where recall that $\boldsymbol{\theta}_1=\boldsymbol{\beta}$. Then, given any fixed $\boldsymbol{\beta}$, the existence of the closed-form MLEs of $\boldsymbol{\theta}_2$ is guaranteed; see Equations \ref{mle1}--\ref{mle3}. Assuming a fixed $\boldsymbol{\theta}_1$, the iterative updating procedure $A(\boldsymbol{\theta}_2; \boldsymbol{\theta}_1)$ is, thus, a standard EM algorithm. Hence, the convergence of $A(\boldsymbol{\theta}_2; \boldsymbol{\theta}_1)$ and the monotonicity of $\mathbf{Q}$ (i.e., $\mathbf{Q}(\boldsymbol{\theta}_2^{(t+1)}; \boldsymbol{\theta}_1) > \mathbf{Q}(\boldsymbol{\theta}_2^{(t)}; \boldsymbol{\theta}_1)$) is guaranteed by \citet{Wu}.
$\square$
\end{proof}

\begin{thm}\label{thm4}
All the limit points of the ECM sequence above are stationary points of the observed-data loglikelihood $\ell_o(\boldsymbol{\theta})$.
\end{thm}

\begin{proof}
From Theorems \ref{thm2} and \ref{thm3}, we conclude $\ell_c(\boldsymbol{\theta}^{(t+1)}) > \ell_c(\boldsymbol{\theta}^{(t)})$ for every $t$. Since the objective function $\ell_c(\boldsymbol{\theta})$ is a density function belongs to the exponential family, it is jointly continuous and concave in $ \boldsymbol{\theta}$. Hence, the corresponding CM-step always converges to a stationary point. Thus, the result is an immediate consequence of Theorem 4 in \citet{Meng}. $\square$
\end{proof}

In summary, unlike the block-relaxation method, the ECM algorithm is guaranteed to converge to a stationary point of the loglikelihood function, assuming the flare regression model. Due to the fact that the density function of the flare regression model belongs to the exponential family, this limiting point to which the ECM algorithm converged is a local maximum of the likelihood function.

We first generate a dataset of size $n=200$ from an EMG regression model, where $x_{i,1}\equiv 1$, $x_{i,2}\sim\mathcal{N}(0,1)$, $\boldsymbol{\beta}=(-2, 4)^{\top}$, and $\epsilon_i \sim EMG(0, \sigma=0.5, \alpha=0.05) $, $i=1,\ldots,n$.  Thus, $\boldsymbol{\psi}=((-2,4),0.5,0.05)^{\top}$. 
Standard errors for the point estimates are estimated via bootstrap since the EMG regression model does not meet the regularity conditions necessary for Louis' method. The results are displayed in the upper-half of Table \ref{sim1_res}. Estimation of the exponential rate $\alpha$ is fairly precise. However, we receive a slightly biased estimate of $\boldsymbol{\beta}$ and an extremely biased estimate on the variance of the Gaussian part, $\sigma^2$. This result is consistent with the properties of the EMG regression discussed in Section \ref{subsec:flare}.  The fact the $\alpha \sigma < 1$ guarantees the loglikelihood to be strictly concave in $\boldsymbol{\beta}$ and $\alpha$ results in relatively accurate estimates of $\boldsymbol{\beta}$ and $\alpha$. However, due to the lack of concavity of the loglikelihood in $\sigma$, its resulting estimate is noticeably less accurate. 

\begin{table}
\caption{\label{sim1_res} Parameter estimates for two regression models examples}
\centering
\hskip-1.0cm\begin{tabular}{cccc}
\toprule
\multirow{2}{*}{\textbf{Parameter}} & \multirow{2}{*}{\textbf{Estimate}} & \textbf{Estimated SE} & \textbf{Estimated SE}\\
 &  & \textbf{(Bootstrap)} & \textbf{(Louis' Method)}\\
\midrule
\multicolumn{4}{c}{EMG Regression Model} \\
\midrule
 $\beta_0$  & $-1.4442$ & $0.1650$ & N/A\\
 $\beta_1$  & $3.2498$ & $0.1808$ & N/A\\
 $\sigma^2$  &$4.0191$ & $0.0864$ & N/A\\
  $\alpha$  & $0.0468$ & $0.0014$ & N/A\\
\midrule
\multicolumn{4}{c}{Flare Regression Model} \\
\midrule
 $\lambda$ & $0.6373 $ & $0.0163 $ & $0.0258 $\\
 $\beta_0$  & $-2.0112 $ & $0.0191 $ & $0.0250 $\\
 $\beta_1$  & $4.0392 $ & $0.0213 $ & $0.0257 $\\
 $\sigma^2$  & $0.2358 $ & $0.0193 $ & $0.0337 $\\
  $\alpha$  & $0.0491 $ & $0.0028 $ & $0.0047 $\\
\bottomrule
\end{tabular}
\end{table}

We next generate a dataset of size $n=200$ from a flare regression model with the exact same conditions as in the preceding EMG regression example, but where the error structure is now $\epsilon_i \sim \lambda \mathcal{N}(0, \sigma^2) + (1-\lambda) Exp(\alpha) $, $i=1, \ldots, n$, such that $\lambda=0.6$.  Thus, $\boldsymbol{\theta}=(0.6,(-2,4),0.5,0.05)^{\top}$. Since the flare regression model meets the regularity conditions necessary for Louis' method, we will estimate standard errors of point estimates through both bootstrapping and Louis' method. After fitting the simulated data with the ECM algorithm, the results are reported in the lower-half of Table \ref{sim1_res}. Unlike the EMG regression estimates, all of the parameter estimates for the present flare regression model appear sufficiently accurate. This result is consistent with the global convergence of the ECM algorithm established in Section \ref{subsec:ECM}. Moreover, the two different methods for estimating standard errors yield results that are roughly the same order of magnitude. Model-based clustering is then performed on the generated data assuming the flare regression model. Using the cut-off probability $p^*=0.80$, the posterior membership probabilities estimated by the ECM algorithm performed an excellent task in identifying observations belonging to the exponential component. For the $200$ generated observations, $126$ were generated from the Gaussian distribution and $74$ from the exponential distribution.  The clustering process yielded $132$ observations classified to the Gaussian component and the remaining $68$ to the exponential component. The $68$ observations classified to ``exponential'' were indeed generated from the exponential distribution, whereas only $6$ out of $132$ observations classified to ``Gaussian" were in fact generated from the exponential distribution. Please refer to Web Table 0 for additional results and Web Figure 2 for scatterplots of these simulations.

EM algorithms are known to be sensitive to starting values \citep{Biernacki,Karlis}. We perform a limited numerical study to assess the robustness of our ECM algorithm under various starting values. We proceed by generating a dataset of size $n=1000$ from a flare regression model with the parameter $\boldsymbol{\theta}=(0.5,(1,4),0.5,0.05)^{\top}$. For each generated sample, we estimate the parameters of the flare model by implementing the ECM algorithm with starting values generated as follows: $\lambda \sim Unif(0, 1)$, $\beta_i \sim \mathcal{N}(0, 1)$ for $i=0, 1$, $\sigma \sim Unif(0, 5)$, and $\alpha \sim Unif(0, 1)$. All of the parameter estimates are sufficiently accurate under these sets of random starting values. Please refer to Web Table 12 for the calculated RMSEs and mean biases for this part of the study.

\subsection{Performance of the ECM Algorithm}

We next perform a larger simulation study to assess the performance and robustness of the ECM algorithm for the flare regression model. This involves the calculation of root-mean-square errors (RMSEs) and biases. Three other candidate models, including the EMG regression model, are also estimated using the simulated data. Bayesian information criterion \citep[BIC;][]{schw} values are calculated to characterize the performance of the flare regression model and its corresponding ECM algorithm relative to the estimates obtained from the other candidate models.

We consider two different conditions for the regression predictors: one with a single predictor and one with two predictors. The predictors under each condition are generated as $x_{i,j} \sim Unif[-10,10]$, $i=1, \ldots, n$, $j =2,3$, and, again, setting $x_{i,1}\equiv 1$. 
We further consider three different scenarios on the mixture components for the errors: well-separated components, moderately-separated components, and overlapping components. For each scenario, we randomly generated $B = 1000$ Monte Carlo samples for each of the sample sizes $n\in\{ 100, 500,1000\}$. The explicit parameter settings for all 12 data-generating models are given in Table \ref{tab:models}.  Please also refer to Web Figures 3--6 in the Supporting Information for visualizations of these simulation settings.





\begin{table}
\caption{\label{tab:models}Parameter settings for the simulation regarding the flare regression model}
\centering
\begin{tabular}{cccccc}
\toprule
\textbf{Setting} & \textbf{Component Structure} & $(\lambda, 1-\lambda)$ & $\boldsymbol{\beta}$&  $\sigma$ & $\alpha$\\
\midrule
M1 & Well-Separated & $(0.333, 0.667)$ & $(9, 3)$ & $0.5$ & $0.05$\\
M2 &Moderately-Separated & $(0.333, 0.667)$ &  $(9, 3)$ & $0.5$ & $0.17$\\
M3 &Overlapping & $(0.333, 0.667)$ & $(9, 3)$ & $0.5$ & $0.5$\\
M4 & Well-Separated & $(0.9, 0.1)$ & $(9, 3)$ & $0.5$ & $0.05$\\
M5 & Moderately-Separated & $(0.9, 0.1)$ & $(9, 3)$ & $0.5$ & $0.17$\\
M6 & Overlapping & $(0.9, 0.1)$ & $(9, 3)$ & $0.5$ & $0.5$\\
M7 & Well-Separated & $(0.5, 0.5)$ & $(-2, 1, 13)$ & $0.5$ & $0.04$\\
M8 & Moderately-Separated & $(0.5, 0.5)$ & $(-2, 1, 13)$ & $0.5$ & $0.2$\\
M9 & Overlapping & $(0.5, 0.5)$ & $(-2, 1, 13)$ & $0.5$ & $0.5$\\
M10 & Well-Separated & $(0.9, 0.1)$ & $(-2, 1, 13)$ & $0.5$ & $0.04$\\
M11 & Moderately-Separated & $(0.9, 0.1)$ & $(-2, 1, 13)$ & $0.5$ & $0.2$\\
M12 & Overlapping & $(0.9, 0.1)$ & $(-2, 1, 13)$ & $0.5$ & $0.5$\\
\bottomrule
\end{tabular}
\end{table}


Tables of the RMSEs and biases are given in Web Tables 2--5 in the Supporting Information.  Visualizations of these tabulated results are also given in Web Figures 7--18 of the Supporting Information. From these results, we can summarize some of the behavior exhibited by the RMSEs and biases across the 12 simulation models.

In 11 out of the 12 simulation settings, both the calculated RMSEs and mean biases show sufficiently low magnitude orders in absolute value, with setting M12 being the only exception. This demonstrates satisfactory precision of the ECM algorithm. The imprecise parameter estimates from setting M12 occur due to the fact that only a small proportion of data  were generated from the exponential component ($\lambda=0.9$) and that the exponential rate was set to be a large value ($\alpha=0.5$). A small mixing proportion for the exponential component, along with this larger exponential rate, will obfuscate the identifiability of the mixture model. Like traditional EM algorithms, estimating with ECM algorithms suffer when faced with model identifiability problems.

In most of the simulation settings, both the calculated RMSEs and mean biases noticeably decrease when the sample size increases from $n=100$ to $n=1000$. As expected, this behavior shows that our ECM algorithm, like other optimization algorithms, tends to perform better as the sample size becomes larger. 

Finally, in most of the simulation settings, the ECM algorithm outputs estimates with lower RMSE and mean bias values under the simulation scenario with well-separated components, and outputs estimation results with higher RMSE and mean bias values under the simulation scenarios with overlapping components. This shows that the ECM algorithm consistently produces more precise estimates when data arise from a mixture with well-separated components. Similar to the reason noted earlier about the subpar performance using data generated from setting M12, a lack of model identifiability emerges when the mixture components heavily overlap with each other. 

\subsection{Broader Model Comparison Study}\label{subsec:model}
We further examine the efficacy of the flare regression model by fitting three other models to the simulated data: the EMG regression model, the classic linear regression model, and a two-component mixture-of-linear-regressions model. The EMG regression model and its corresponding block-relaxation algorithm are as presented in Sections \ref{s:model} and \ref{s:inf}. The classic linear regression model is just the model in 
(\ref{linreg}), but where $\epsilon_i \sim N(0, \sigma^2)$, for $i=1, \ldots, n$. Here, the parameter vector of interest is $(\boldsymbol{\beta}^{\top}$, $\sigma^2)^{\top}$, which is estimated by ordinary least squares. For the two-component mixture of linear regressions, we have:
\[ y_i =
    \begin{cases}
      \mathbf{x}_i \boldsymbol{\beta}_1 + \epsilon_{i1}, & \text{with probability}\ \lambda; \\
      \mathbf{x}_i \boldsymbol{\beta}_2 + \epsilon_{i2}, & \text{with probability}\ 1-\lambda,
    \end{cases} \]
where the $ \epsilon_{ij} \sim N(0, \sigma_j^2)$ are (conditionally) $iid$, $i=1, \ldots, n$ and $j=1, 2$.  In this mixture model, the parameter vector of interest is $(\lambda$, $\boldsymbol{\beta}_1^{\top}$, $\boldsymbol{\beta}_1^{\top}$, $\sigma_1^2$, $\sigma_1^2)^{\top}$, whose closed-form MLEs can easily be derived using a standard EM algorithm \citep{deVeaux1989}. This EM algorithm is implemented by the  \verb|regmixEM()| function in the \verb|R| package \verb|mixtools| \citep{mixtools}.

On average, over 90\% of the time the flare regression model outperforms the other candidate models in terms of having the lowest BIC values. Please see Web Table 1 of the Supporting Information for detailed percentages of the lowest BIC values from all the candidate models estimated using the data generated from each simulation setting.  Overall, this demonstrates the strong performance of the flare regression model, especially in the context of data that one might typically consider modeling with EMG regression.

\section{Application: Human Aiming Performance Data}
\label{s:data}

\subsection{Data Description and Model Settings}

We now analyze data from the field study by \citet{Chapuis}, which was produced by unobtrusively collecting mouse input and corresponding GUI data from 24 users over several months. The dataset consists of more than 2 million movements. Many variables of interest were collected, including time, cursor position, mouse movements, mouse and button events (click, drag, long click), type and properties of the selected target, such as size and role of the target (e.g., resizing button, edge of a window), as well as information regarding the system used (which input device, desktop/laptop, Operating System). In this work, we used only information on movement time, distance to the target, and target size, which is consistent with applying Fitts' model. 

 The theoretical model is given by $y = \beta_0 + \beta_1 x$, where $y = t_e/1000$ (converting milliseconds to seconds) and $x = \log_2 \bigg(1 + \frac{dist}{\text{min}(w_t,h_t)}\bigg)$. The variable $x$ is considered a difficulty measure, whose units are in bits. Typically in controlled studies with computer mice, $\beta_0 \in [-0.1,0.1]$ and $\beta_1 \in [0.1,0.2]$, where $\beta_0$ is in seconds, and $\beta_1$ is in seconds/bit. 
 Compared to data typically collected in controlled studies, these data display noticeable positive skewness because users do not routinely try to minimize their task completion time.  Figure \ref{scatterdata} is a scatterplot of the data from one user in our data.  Notice the variability and considerable positive skewness in the task completion times. Unlike aiming data collected in controlled studies, a linear regression assuming zero-centered Gaussian noise is not an appropriate model for the present data. See Figure \ref{qq} for a plot the residuals versus the fitted values and the corresponding quantile-quantile plot when fitting a simple linear regression model to the data collected from the same user in Figure \ref{scatterdata}.
 
 \begin{figure}
\begin{center}
        \includegraphics[scale=1.0]{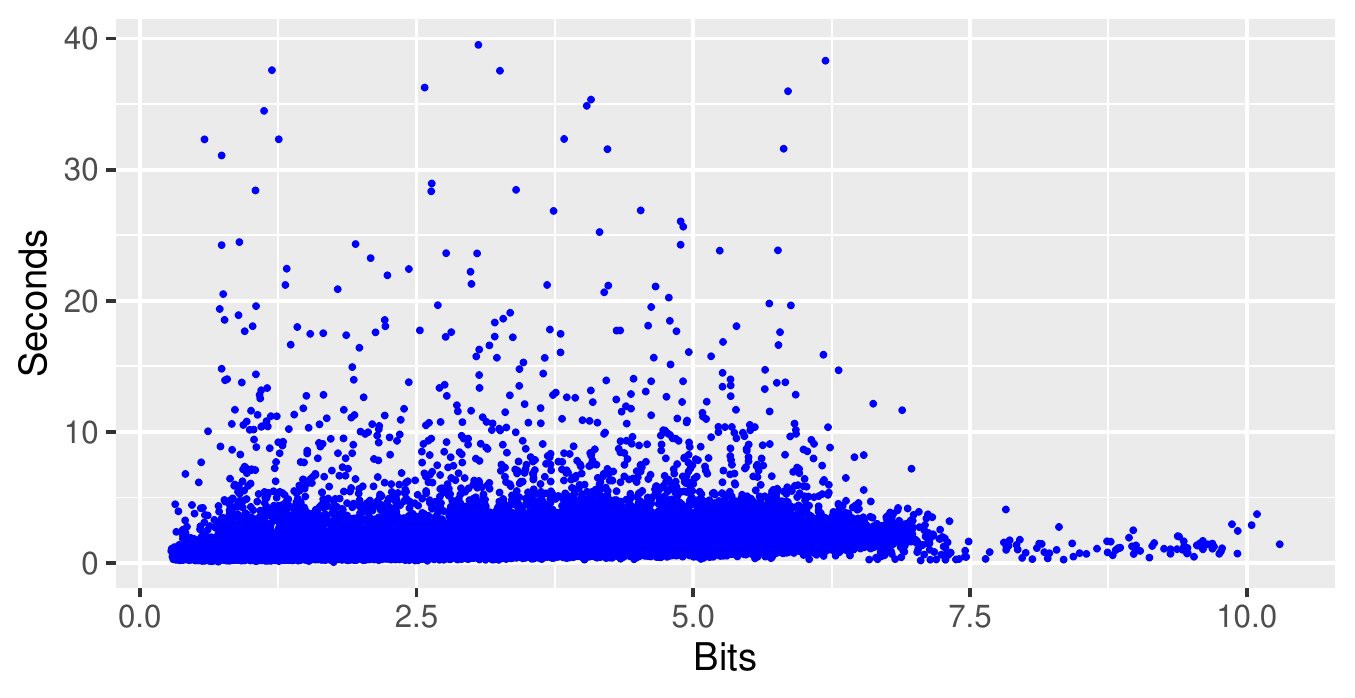}
   \caption{Scatterplot of real-world aiming performance from a user (User 1 in the Supporting Information)}\label{scatterdata}
\end{center}
\end{figure}

\begin{figure}
\begin{center}
    \subfigure[]
    {
        \includegraphics[scale=1]{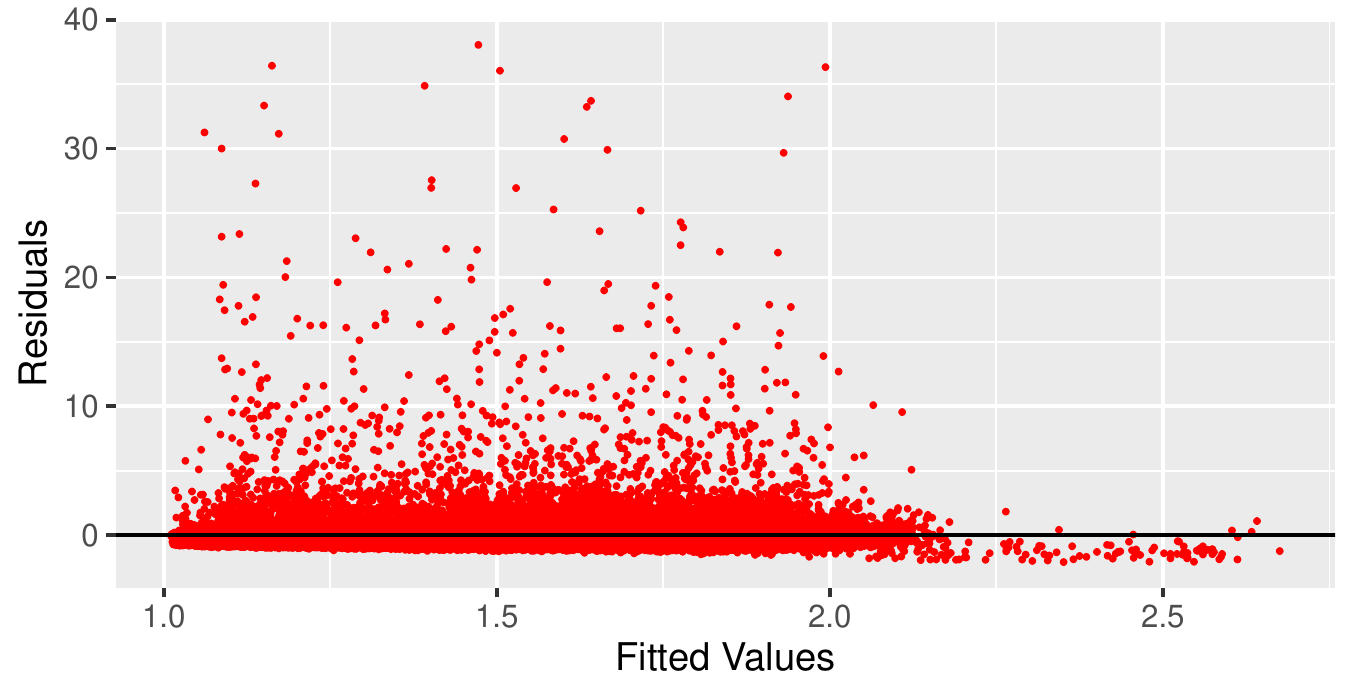}
      
    }
    \subfigure[]
    {
        \includegraphics[scale=1]{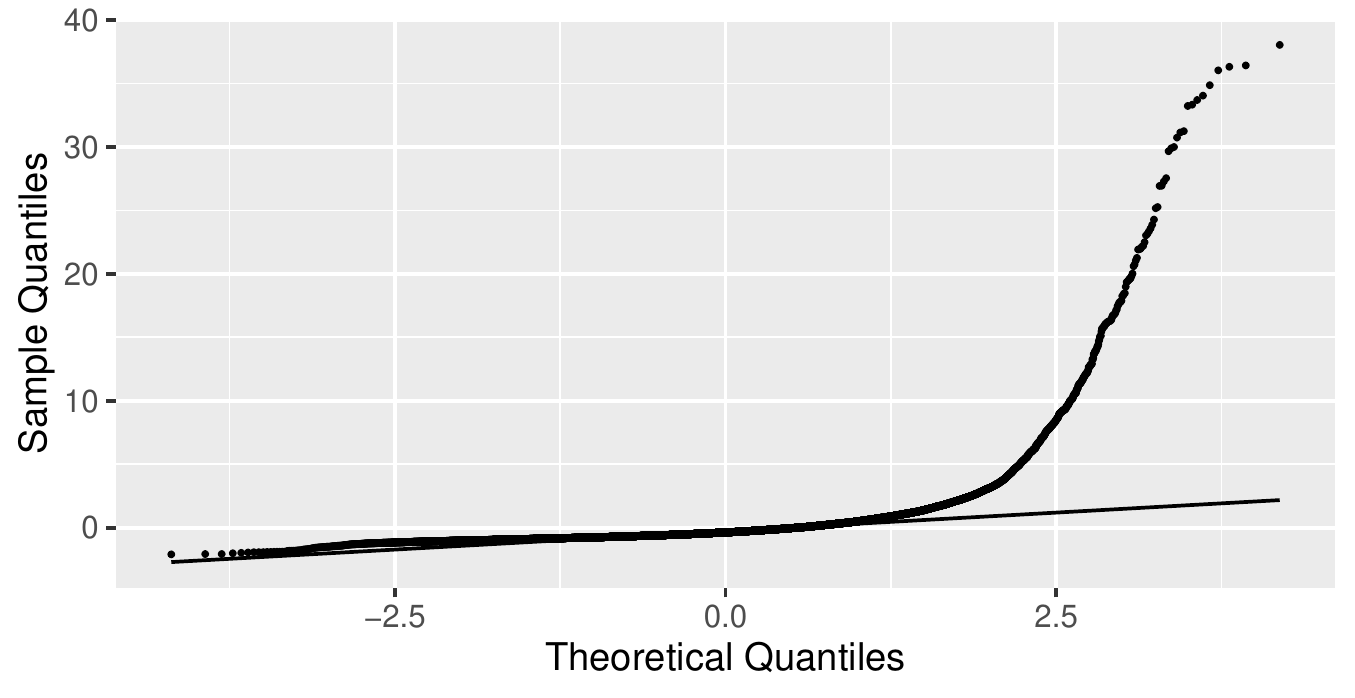}
        
    }
    
   \caption{(a) Scatterplot of the residuals versus the fitted values from a classic linear regression fit and (b) the corresponding quantile-quantile plot   }\label{qq}
\end{center}
\end{figure}

In \citet{Julien2}, the EMG regression model was estimated with a very small subset of the ``in the wild" data. Compared to classic linear regression with Gaussian errors, the estimated EMG regression parameters fall within the typical range of those for controlled experiments, and the fitted line matches well with the idea of minimum movement time.  We extend this previous work by fitting and comparing the four models used in the simulation study discussed in Section \ref{subsec:model}. Additionally, instead of only the small subset analyzed in \citet{Julien2}, we use the entire ``in the wild" dataset when estimating the four candidate models for each of the 24 users.

\subsection{Data Truncation and Estimation Results}

Besides the characteristic positive skewness of the ``in the wild" data, technical difficulties associated with trajectory segmentation frequently produce outliers. To obtain informative estimates, outliers produced by technical errors should be eliminated. However, there is no definitive indicator as to when an observation is an outlier.  Thus, four different cut-off thresholds are investigated: $T = 10s$, $T = 20s$, $T = 30s$, and $T = 40s$. When we set a fixed cut-off threshold, only observations with response time $y$ less than the threshold will be considered (i.e., $y_i \leq T$). As the cut-off threshold increases, more extreme values of long reaction times are present in the corresponding truncated data. 

After fitting four candidate models discussed in Section \ref{subsec:model}, we find the EMG regression model and the flare regression model consistently outperform the other two candidate regression models (i.e., simple linear regression with Gaussian errors and the two-component mixture of linear regressions) by producing significantly lower BIC values under all four cut-off thresholds. As we increase the cut-off threshold, however, more extreme values are naturally present, and the flare regression model tends to perform better than the EMG regression model.  Figure \ref{BIC} provides a visualization of how the BIC values for each participating user change as the cut-off threshold increases. In terms of their BIC values, blue cells correspond to the EMG regression being a better fit, while the red cells correspond to the flare regression model being a better fit.  As the threshold increases, more red cells appear in the figure.  Thus, we see the ability of the flare regression model relative to better characterize more extreme values relative to the EMG regression model for these aiming performance data.

\begin{figure}
\begin{center}
        \includegraphics[scale=1.2]{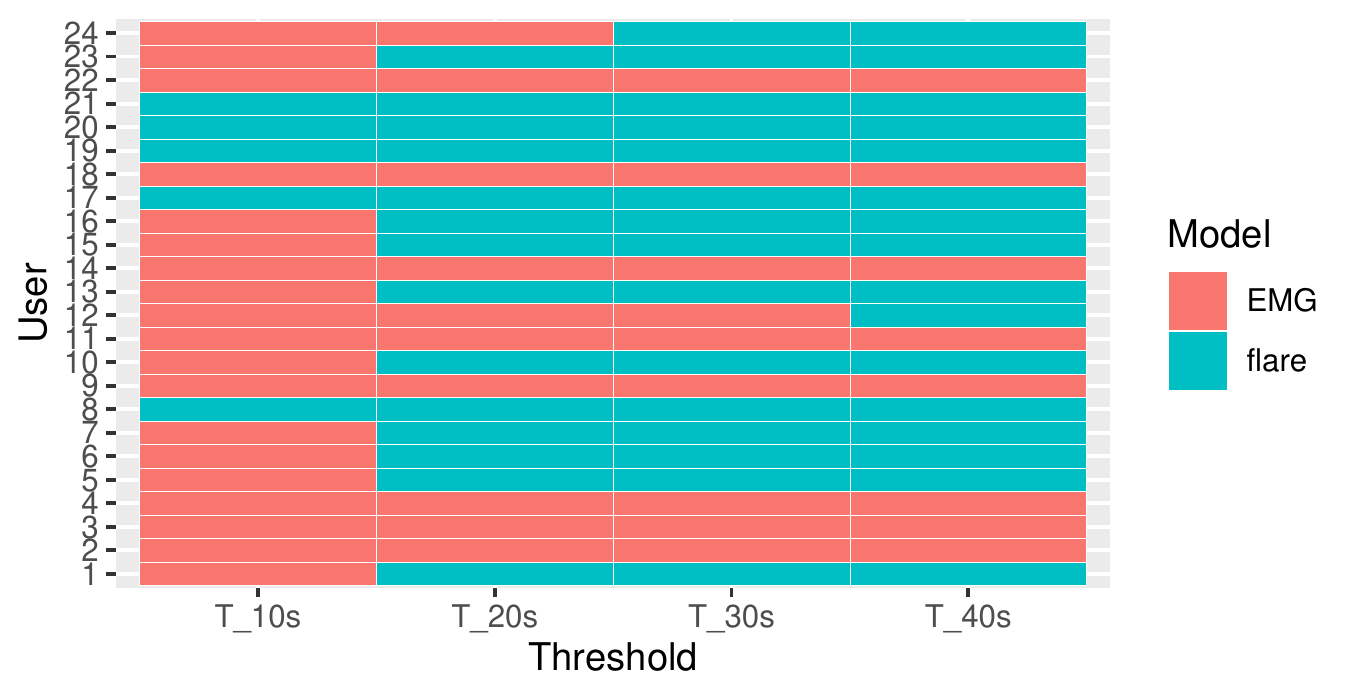}
   \caption{BIC comparisons for the EMG regression and flare regression model fits with different cut-off values. Orange cells indicate lower BIC values for the EMG model; green cells indicate lower BIC values for the flare regression model}\label{BIC}
\end{center}
\end{figure}

After balancing the need to eliminate outliers produced by possible technical errors and the necessity of preserving observations with long movement times, the cut-off threshold $T = 40s$ is selected. When the entire dataset is truncated using $T = 40s$, the flare regression model outperforms all the other candidate models with much lower BIC values for 16 out of the 24 participating users. Moreover, we receive similar parameter estimates after fitting the four candidate models to the data from each user. Similar parameter estimates show users tend to have similar movement times while completing the aiming tasks. Exact parameter estimates obtained for the four candidate models  using the truncated data with threshold $T = 40s$ are in Web Tables 7--10 of the Supporting Information and BIC comparisons are in Web Table 6 of the Supporting Information. Besides the BIC values, linear regression yields parameter estimates outside the typical intervals for $\beta_0$ and $\beta_1$ in controlled studies. The EMG and flare regression models, on the other hand, yield parameter estimates within the typical intervals. However, the two models behave differently: the EMG regression model tends to yield more estimates for the intercept inside the typical interval, whereas the flare regression model tends to yield more estimates for the slope inside the typical interval. A visualization for this comparison is in Web Figure 19 of the Supporting Information. Comparing to the intercept, researchers consider the slope to be a more informative parameter when measuring  movement difficulty \citep{Zhai,PLoSone}.

\subsection{Classification and Interpretations}

As noted in Section \ref{sec:intro}, both the EMG regression model and our flare regression model are able to effectively handle data with positive residuals relative to their underlying mean trend, which is a prominent feature of this ``in the wild" data.  However, as demonstrated in Section \ref{s:inf}, 
we can further perform model-based clustering on this ``in the wild" data based on $Z_i^{(\infty)}$, the posterior membership probabilities.  Those observations classified to the Gaussian component would represent the typical movement times of individuals in a controlled study, consistent with Fitts' law.  Those observations classified  to the exponential component would represent where a user is not trying to maximize their performance as well as any possible outliers that have not been removed due to the truncating strategy employed earlier.

For example, Figure \ref{scatterpost} is a scatterplot for the same user in Figure \ref{scatterdata} after fitting the flare regression model. In this figure, the flare regression model fit has been overlaid along with each observation color-coded according to their component membership based on their maximum posterior membership probability (i.e., the cut-off probability $p^*$ is set to be $0.50$).  We have been able to effectively characterize the regions where the user has almost certainly not been performing in an optimal capacity for the aiming task.  Moreover, this region could still include some outlying values associated with trajectory segmentation.

Note that we have done a hard classification based on an observation's posterior membership probabilities.  However, the noticeable delineation between the Gaussian component and exponential component, as seen in Figure \ref{scatterpost}, appears in each user's fit. Further examination shows that the posterior membership probabilities unsurprisingly hover around 0.50 for the two components in this region as this is where the two components have more substantial overlap.  If interested, one could apply a color gradient relative to the membership probabilities to visualize the uncertainty of assignment to one component over the other, thus providing a more nuanced interpretation about the user's performance.

\begin{figure}
\begin{center}
        \includegraphics[scale=1.2]{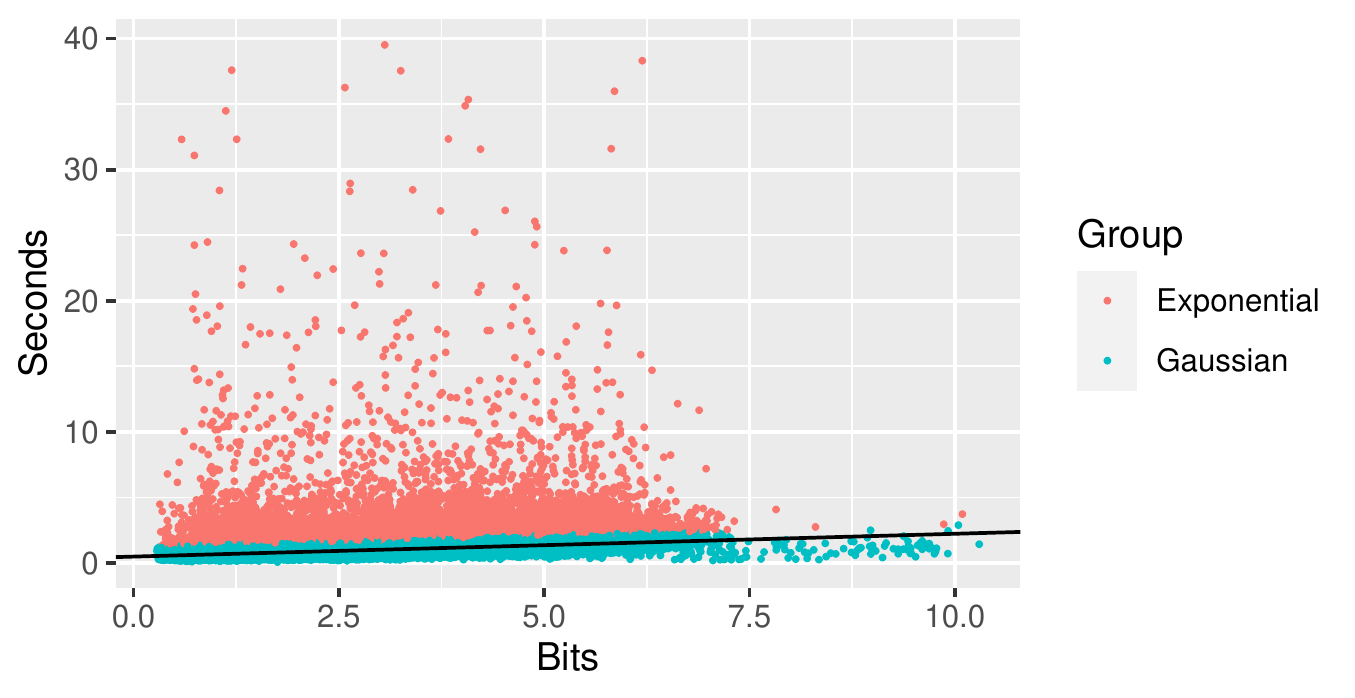}
   \caption{Scatterplot of the user's data in Figure \ref{scatterdata} (User 1), but with the flare regression model fit ($\hat{\beta}_0 = 0.49$, $\hat{\beta}_1 = 0.17$) overlaid along with each observation color-coded according to their component membership based on their maximum posterior membership probabilities (please refer to the first row of Web Table 7 in the Supporting Information for detailed estimation results) }\label{scatterpost}
\end{center}
\end{figure}

The ``in the wild" data distinguishes itself from other data collected in controlled studies by displaying observations with extremely long task completion times. Hence, identifying outliers is essential to obtain informative  results. In this study, we proceed with a conservative approach by selecting a uniform cut-off threshold ($T = 40s$), and drop all of the observations that exceed this threshold. Researchers may be interested in finding alternative methods to remove outliers. In the flare regression model, observations with exceedingly long duration times are highly likely to be classified to the exponential component. This clustering feature allows a framework for outlier removal and for HCI researchers to focus on those observations that are more consistent with what is typically observed in controlled studies. Subsequently, even more candidate models could be investigated for characterizing the remaining observations in the data, thus allowing more nuanced comparisons between ``in the wild" data and data from controlled studies for HCI research. 

\subsection{System Running Time}

When the sample size gets (excessively) large, system running time is used in assessing the performance of algorithms. All of the algorithms used in this study were implemented in \verb|R|. The mean sample size of the user datasets is about 19667. The mean system running time of the ECM algorithm is 8.8502 seconds,whereas, the mean system running time of the block-relaxation algorithm is 193.6914 seconds. For every user, the system elapsed time of the ECM algorithm is significantly shorter than the block-relaxation algorithm. Web Table 11 in the Supporting Information summarizes individual user's sample size and system elapsed times across the four candidate models. Theoretical and technical reasons behind this empirical finding remain a potential future direction of research.

\section{Concluding Remarks}
\label{s:discuss}
The EMG distribution is a practical model applied in various fields when researchers encounter positively skewed data, especially when it involves timing studies of tasks with human subjects like the human aiming performance data that motivated this study. This paper addressed some of the computational challenges in estimating an EMG regression model with multiple predictors by developing an iterative block-relaxation algorithm. Even though some concave properties of the EMG regression are proved, the fact the the EMG distribution is not a member of the exponential family prevents us from guaranteeing global concavity of the loglikelihood and convergence  of the block-relaxation algorithm. Alternatively, we introduced our novel flare regression model consisting of a two-component mixture structure on the errors, consisting of a Gaussian component and an exponential component. We developed an ECM algorithm for estimation, which unlike the block-relaxation algorithm, is guaranteed to converge to a local maximum of its likelihood function. After obtaining point estimates of the flare regression model, we briefly addressed the calculation of estimated standard errors for the parameter estimates. 

Both the extensive simulation study and the analysis of the human aiming performance data showed significant advantages of the flare regression model over the EMG regression model and other existing regression models. Not only is the flare regression model fit typically better than the EMG regression model fit (in terms of BIC values), the former also provides us with additional insight into different performance regions in the human aiming task.  Moreover, a timing comparison between the block-relaxation method for the EMG regression and the ECM algorithm for the flare regression shows superior performance for the latter. Overall, we have shown that the flare regression model is highly efficacious as a way for characterizing the human aiming performance data analyzed in this work.  

There are various avenues of future research to expand the work presented here.  For example, our model is proposed as the mixture alternative to the EMG regression model.  Of course, other skewed distributions could be explored for the second component of our model, and there may be some sort of optimality criterion that could be employed for identifying such a distribution.  But given the recent attention of EMG models in the HCI literature, and more generally their prevalence in reaction times applications, it makes sense that we proposed a mixture model analogue to that model.  This framework of comparing a mixture model whose components comprise an established convolution model could also be employed for other applications, such as the Voigt profile used in spectroscopy, which is given by a convolution of a Gaussian distribution and a Cauchy distribution. Another extension, as noted in the analysis of Section \ref{s:data}, is to better characterize subject-to-subject variability in terms of performance on this task.  Incorporation of random effects to allow for such subject heterogeneity would likely provide an even more informative model. Thus, generalizing both the EMG and flare regression models by incorporating random effects, and then comparing the results, would be an informative direction for future research.

\newpage

\appendix

\section{Appendix: Algorithms}

\begin{algorithm}[H]\label{alg1}
	\caption{Block-Relaxation Method for EMG Regression} 
	\hspace*{\algorithmicindent} \textbf{Input:}  $\mathbf{X}$ (matrix of predictors), $\mathbf{y}$ (response vector)\\
    \hspace*{\algorithmicindent} \textbf{Output:} Final estimate $\widehat{\boldsymbol{\psi}}$ for $\boldsymbol{\psi} = ( \boldsymbol{\beta}^{\top}, \sigma^2, \alpha)$
	\begin{algorithmic}[1]
	
	 \State Initialize the iteration $t=0$; set the difference $\textit{diff} = 1$
	 
	 \State Initialize the method by selecting starting values $\boldsymbol{\psi}^{(0)} = ( {\boldsymbol{\beta}^{(0)}}^{\top}, {\sigma^2}^{(0)},  \alpha^{(0)} )^{\top}$
	 
	  \While{$\textit{diff} > \epsilon$}{
	  
	  Update $\boldsymbol{\psi}_1^{(t+1)} =  \underset{\boldsymbol{\psi}_1}{\arg\max}\ \mathbf{Q}(\boldsymbol{\psi}_1; \boldsymbol{\psi}_2^{(t)}) $
	  
	  Update $\boldsymbol{\psi}_2^{(t+1)} = \underset{\boldsymbol{\psi}_2}{\arg\max} \ \mathbf{Q}(\boldsymbol{\psi}_2; \boldsymbol{\psi}_1^{(t+1)}) $
	 
	  Update the difference: $\textit{diff} \longleftarrow |\mathbf{Q}(\boldsymbol{\psi}^{(t+1)})-\mathbf{Q}(\boldsymbol{\psi}^{(t)})| $
	  	          
	  Update $\boldsymbol{\psi}^{(t)} \longleftarrow  \boldsymbol{\psi}^{(t+1)}$

	  $t \longleftarrow t+1$
       
    \EndWhile}
	 
    \State  Output $\widehat{\boldsymbol{\psi}} = \boldsymbol{\psi}^{(t)}$ 
    
	\end{algorithmic} 
\end{algorithm}

\

\

\

\

\

\

\

\

\

\

\begin{algorithm}[H]\label{alg2}
	\caption{ECM Algorithm for Flare Regression Model} 
	\hspace*{\algorithmicindent} \textbf{Input:}  $\mathbf{X}$ (matrix of predictors), $\mathbf{y}$ (response vector)\\
    \hspace*{\algorithmicindent} \textbf{Output:} Final estimate $\widehat{\boldsymbol{\theta}}$ for $\boldsymbol{\theta} = (\lambda,  \boldsymbol{\beta}^{\top}, \sigma^2, \alpha)^{\top}$
	\begin{algorithmic}[1]
	
	 \State Initialize the iteration $t=0$
	 
    \State Initialize the estimation by selecting starting values $\boldsymbol{\theta}^{(0)} = (\lambda^{(0)}, {\boldsymbol{\beta}^{(0)}}^{\top}, {\sigma^2}^{(0)}, \alpha^{(0)})^{\top} $
    
    \State Estimate the initial hidden variable:     $  Z_i^{(0)} = \frac{\frac{\lambda^{(0)}}{2\pi\sigma^{2(0)}} \exp \Big\{ -\frac{1}{2\sigma^{2(0)}} (y_i - \mathbf{x}_i^{\top} \boldsymbol{\beta}^{(0)})^2  \Big\}}{f(y_i; \mathbf{x}_i, \boldsymbol{\psi}^{(0)})} $
    
    \State Initialize the difference $\textit{diff} = 1$; update the objective function $m(\boldsymbol{\beta}) = \mathbf{Q}(\boldsymbol{\beta}; \boldsymbol{\theta}^{(0)})$

    \While{$\textit{diff} > \epsilon$}{
    
         Estimate parameters in $\boldsymbol{\theta}_2$: 
        $ \boldsymbol{\beta}^{(t+1)} = \boldsymbol{\beta}^{(t)} -  \left[\frac{d^2m}{d\boldsymbol{\beta}^2}\right] ^{-1}\Biggl|_{\boldsymbol{\beta}=\boldsymbol{\beta}^{(t)}}   \frac{dm}{d\boldsymbol{\beta}}\Biggl|_{\boldsymbol{\beta}=\boldsymbol{\beta}^{(t)}}$
        
        Set $\boldsymbol{\theta}^{(t+1/2)} = (\lambda^{(t)}, {\boldsymbol{\beta}^{(t+1)}}^{\top}, {\sigma^2}^{(t)}, \alpha^{(t)})^{\top}$
        
        Re-estimate the hidden variable:  $Z_i^{(t+1/2)} = \frac{\frac{\lambda^{(t)}}{2\pi\sigma^{2(t)}} \exp \Big\{ -\frac{1}{2\sigma^{2(t)}} (y_i - \mathbf{x}_i^{\top} \boldsymbol{\beta}^{(t+1)})^2  \Big\}}{f(y_i; \mathbf{x}_i, \boldsymbol{\theta}^{(t+1/2)})} $
        
        Estimate parameters in $\boldsymbol{\theta}_2$: 
        \[ \lambda^{(t+1)} = \frac{1}{n} \sum_{i=1}^n Z_i^{(t+1/2)} \]
        
        \[      \sigma^{2(t+1)} = \frac{\sum_{i=1}^n Z_i^{(t+1/2)}(y_i - \mathbf{x}_i^{\top} \boldsymbol{\beta}^{(t+1)})^2}{\sum_{i=1}^n Z_i^{(t+1/2)}} \]
        
        \[     \alpha^{(t+1)} = \frac{\sum_{i=1}^n (1-Z_i^{(t+1/2)})}{\sum_{i=1}^n (1-Z_i^{(t+1/2)}) (y_i - \mathbf{x}_i^{\top} \boldsymbol{\beta}^{(t+1)})} \]
        
        Set $\boldsymbol{\theta}^{(t+1)} = (\lambda^{(t+1)}, {\boldsymbol{\beta}^{(t+1)}}^{\top}, {\sigma^2}^{(t+1)}, \alpha^{(t+1)})^{\top}$
        
        Update the difference: $\textit{diff} \longleftarrow ||{\boldsymbol{\theta}^{(t+1)}}  - {\boldsymbol{\theta}^{(t)}}||_{\infty} $
        
        Re-estimate the hidden variable: $Z_i^{(t+1)} = \frac{\frac{\lambda^{(t+1)}}{2\pi\sigma^{2(t+1)}} \exp \Big\{ -\frac{1}{2\sigma^{2(t+1)}} (y_i - \mathbf{x}_i^{\top} \boldsymbol{\beta}^{(t+1)})^2  \Big\}}{f(y_i; \mathbf{x}_i, \boldsymbol{\theta}^{(t+1)})} $
        
        Update $\boldsymbol{\theta}^{(t)} \longleftarrow  \boldsymbol{\theta}^{(t+1)}$, $t \longleftarrow t+1$
       
    \EndWhile}
    
   \State  Output $\widehat{\boldsymbol{\theta}} = \boldsymbol{\theta}^{(t)}$ 
    
	\end{algorithmic} 
\end{algorithm}

\bibliographystyle{rss}
\bibliography{example}
\end{document}